\documentclass[12pt]{article}

\pdfoutput=1
\usepackage[pdftex]{color}
\usepackage[pdftex]{graphicx}
\usepackage{booktabs}
\usepackage{amsmath,bm}
\usepackage{amsthm}
\usepackage{amssymb}
\newtheorem{proposition}{Proposition}[section]
\newtheorem{lemma}{Lemma}[section]
\newtheorem{theorem}{Theorem}[section]

\newtheorem{definition}{Definition}[section]
\title{\Large Median regression with differential privacy
\footnote{Supported
 by NNSF of China (11671290)}}
\author{E Chen \textsuperscript{1}, Ying Miao\textsuperscript{2} and Yu Tang\textsuperscript{1} \\
\textsuperscript{1} School of Mathematical Sciences, Soochow University \\
\textsuperscript{2} Faculty of Engineering, Information and Systems, University of Tsukuba \\
}
\date{Apr. 25, 2020}
\begin{document}
\newcommand{\rev}[1]{{\color{red} #1}}
%\renewcommand{\listfigurename}{插图目录}
%\renewcommand{\listtablename}{表格目录}
%\renewcommand{\abstractname}{摘要}
%\renewcommand{\contentsname}{目录}
%\renewcommand{\tablename}{表}
%\renewcommand{\figurename}{图}
%\begin{titlepage}
%\maketitle
%\thispagestyle{empty}
%\tableofcontents
%\thispagestyle{empty}
%\listoffigures
%\thispagestyle{empty}
%\listoftables
%\thispagestyle{empty}
%\newpage
%\end{titlepage}

\maketitle

 \noindent
 \textbf{Abstract:}
 Median regression analysis has robustness properties which make it attractive
compared with regression based on the mean, while differential privacy can protect individual privacy during statistical analysis of certain datasets. In this paper, three privacy
preserving methods are proposed for median regression. The first algorithm is based on
a finite smoothing method, the second provides an iterative way and the last one further
employs the greedy coordinate descent approach. Privacy preserving properties of these
three methods are all proved. Accuracy bound or convergence properties of
these algorithms are also provided. Numerical calculation shows that the first method has better
accuracy than the others when the sample size is small. When the sample size becomes
larger, the first method needs more time while the second method needs less time with well-matched accuracy. For the third method, it costs less time in both
cases, while it highly depends on step size.

 \noindent
\textbf{keywords}: median regression, differential privacy, $l_1$ sensitivity, Laplace mechanism
%locating array, lower bound, mixed-level, simulated annealing

\newpage
                                 %                                            %
\section{Introduction}%这边需要加不少内容
Personal privacy information may be exposed with the unprecedented availability of datasets, so there is increasing requirement that statistical analysis of such datasets should protect individual privacy.
% This requirement pays attention to statistical methods
% which can protect individual privacy and
% guarantee statistical accuracy at the same time.
 As \cite{Dwork2014} describes, differential privacy addresses the paradox of learning nothing about
an individual while learning useful information about a population.
Over the past few years, differential privacy has been investigated in machine learning
\cite{DL2016} and has been applied in the real world, see for example \cite{Google2014}. Recently,
\cite{Tony2019} formulates a general lower bound argument for
minimax risks with differential privacy constraints, and applies this
argument to high-dimensional mean estimation and linear regression
problems.

In this paper, three privacy preserving methods are proposed for median regression, which is a special case of quantile regression.
Quantile regression was first introduced in \cite{QR1978}, which aims to estimate
and conduct inference about conditional quantile functions. In recent years, quantile regression has become a comprehensive
method for statistical analysis of response models and it has been widely used in reality,
such as survival analysis and economics,
see for example, \cite{RML2001}, \cite{Yang1999} and \cite{QR2001}.
The fact that the median regression takes least absolute deviation
as its objective function to estimate parameters has been
known among statisticians \cite{QR1978}.

Denote a dataset of $n$ i.i.d. samples about independent variables as $\bf{X}$ and each observation contains $d$ variables $x_1,x_2,\ldots,x_d$.
In the regression setting,
we assume $Y_i$ is the response for case $i$,  $x_{ij}$ is the value of predictor $j$ for case
$i$, and $\beta_j$ is the regression coefficient corresponding to predictor $j$, where $1 \le i \le n, 1 \le j \le d$.
In this paper, we consider the linear $l_1$ regression problem, i.e.,
minimizing the following function:
\begin{equation}\label{E0}
  F( \mu,{\bm \beta})=\frac{1}{n}\sum^n_{i=1}|r_i( \mu,{\bm \beta})|,
\end{equation}
where
$r_i ( \mu,{\bm \beta})=\mu +{{\bf X}_i}{\bm \beta}-Y_i  ( i=1,2,\cdots,n) $ and ${ \bf X}_i$ represents $i$-th row of {\bf{X}}, and $\bm{\beta} = (\beta_1, \ldots, \beta_d)^\mathrm{T}$. Without loss of generality, assume that
$||Y_i||_1 \leq B$( $B$ is a positive number) and $||{ \bf X}_i||_1 \leq 1$ for  $i=1, \ldots,n$.
In a vector form, $r( \mu,{\bm \beta})=\mu \textbf{1}+{\bf X}{\bm \beta}-{\bf Y}$
represents a set of linear functions in $\mathbb{R}^n$ with ${\bf Y}=(y_1, \ldots, y_n)^\mathrm{T}$, where \textbf{1} is an $n$-dimensional column vector whose all elements are $1$. In addition, the ridge penalized regression is more stable than simple linear regression
 and its objective function can be viewed as minimizing the criterion
\begin{equation}\label{Ridge1}
  L( \mu,{\bm \beta}) = F( \mu,{\bm \beta}) + \frac{\lambda}{2} {\bm \beta}^\mathrm{T} {\bm \beta},
\end{equation}
where $\lambda$ is a fixed regularization parameter.

\section{Backgrounds and definitions}
We consider a dataset $x$ as a collection of observations from a universe $\mathcal{X}$. It is convenient to represent databases by their
histograms: $x \in \mathbb{N}^{|\mathcal{X}|}$, in which each entry $x_i$ represents the number of
elements in the database $x$ of type $i \in \mathcal{X}$.
For example, the universe $\mathcal{X}$ contains $5$ records and we denote them by  $\{1, 2, 3, 4, 5\}$. If a dataset $x$ consists of three records $1,1$ and $4$, we can denote
$x$ as a $5$-dimensional vector $( 2, 0, 0, 1 , 0)$, where the first element is $2$ since record $1$ appears twice.  A $5$-dimensional vector $( 2, 0, 1, 1 , 0)$ represents another dataset $y$ with $4$ records, respectively.

\rev{Differential privacy is based on the neighbourhood of a database, when applying differential privacy into practical use, it is key to define the precise condition under which two databases $x$ and $y$ are considered to be neighbouring. There are two possible choices and thus producing two types of differential privacy, one is called unbounded differential privacy \cite{Dwork2006} and the other is called bounded differential privacy  \cite{Dwork2006C}. Bounded differential privacy assumes that both $x$ and $y$ have the same size $n$ and $y$ can be obtained from $x$ by replacing exactly one record. While unbounded differential privacy does not require $x$ and $y$ have the same fixed size, it holds the view that $y$ can be obtained from $x$ by adding or deleting exactly one record. In this paper, we adopt bounded differential privacy as our choice and use the notation $x \triangledown y$ if $x$ and $y$ are neighboring.}
%For convenience, we use notations similar to \cite{Dwork2014}.
%%\begin{definition}%[Distance between databases]
%%The $l_1$ norm of a database $x$ is denoted as $||x||_1$ and is defined to be:
%%$$
%%||x||_1=\sum_{i=1}^{|\mathcal{X}|}|x_i|
%%$$
%%\end{definition}

%Denote $d_H( x, y)$ as the Hamming distance between two databases $x$ and $y$. That is, $d_H( x, y)$ measures the records differ between $x$ and $y$. Especially, if $d_H( x, y)=1$, we call database $x$ is a neighbourhood of database $y$ and $d_H( x, y)=0$ indicates that databases $x$ and $y$ are same.
%\rev{For example, assume the universe $\mathcal{X}$ contains $5$ records and denote them by  $\{1, 2, 3, 4, 5\}$. If a dataset $x$ consists of three records $1,1$ and $4$, we can denote
% $x$ as a $5$-dimensional vector $( 2, 0, 0, 1 , 0)$, where the first element is $2$ since record $1$ appears twice.  A $5$-dimensional vector $( 2, 0, 1, 1 , 0)$ represents another dataset $y$ with $4$ records, respectively. In this case, it can be verified directly that $d_H( x, y) = 1$ and database $x$ is a neighbourhood of database $y$.
%}

%Notice that $||x||_1$ measure the size of a database $x$ and $||x-y||_1$ measures the records differ between $x$ and $y$.
%Especially, if $||x-y||_1=1$, we call database $x$ is a neighbourhood of database $y$.
\begin{definition}
A randomized algorithm $M$ with domain $\mathbb{N}^{|\mathcal{X}|}$
is $( \epsilon,\delta)$-differentially private if for all $S \subseteq$ Range $( M)$ and for all datasets $x, y \in \mathbb{N}^{|\mathcal{X}|}$ and $x \triangledown y$:
$$
Pr( M(x)\in S) \leq exp( \epsilon)Pr(M(y) \in S)+\delta.
$$
\end{definition}
By intuition, this definition  guarantees that a randomized algorithm behaves similarly on slightly different input datasets,
which achieves the purpose of protecting individual privacy in some sense.
Next, a randomized algorithm named Laplace mechanism, which is an effective method for privacy preserving, will be introduced.
Firstly, we need a concept named $l_1$ sensitivity.
\begin{definition}
The $l_1$ sensitivity of a function $f$:$\mathbb{N}^{|\mathcal{X}|}$ %$\mathbb{N}^{|\mathcal{X}|}$
$\rightarrow \mathbb{R}^k$ is :
$$
\Delta f = \text{max}_{x,y \in \mathbb{N}^{|\mathcal{X}|}, x \triangledown y} ||f( x)-f( y)||_1.
$$

\end{definition}
The $l_1$ sensitivity of a function $f$ captures the magnitude by which a single individual's data can change the function $f$ in the worst case.
It is noteworthy  that $\Delta f$ is an important value in the Laplace mechanism.
\begin{definition}
Given any function $f$:$\mathbb{N}^{|\mathcal{X}|}\rightarrow \mathbb{R}^k$, the Laplace mechanism is defined as:
$$
M_L(x,f(\cdot),\epsilon) = f( x) + ( Y_1,\ldots,Y_k),
$$
where $Y_i \ (i=1,\ldots,k)$ are i.i.d random variables drawn from the Laplace distribution Lap($\frac{\Delta f}{\epsilon})$. The density function of the Laplace distribution ( centered at $0$) Lap($c$) is:
$$
Lap(x|c) = \frac{1}{2c}exp(-\frac{|x|}{c}).
$$
\end{definition}

The following Lemma can be seen in textbooks, see for example Theorem 3.6 of \cite{Dwork2014}.
%According to Theorem 3.6 in \cite{Dwork2014}, we can obtain the following theorem.
\begin{lemma}\label{Lap}
The Laplace mechanism preserves $( \epsilon,0)-$differential privacy.
\end{lemma}
%\begin{proof}
%Proof can be seen in \cite{Dwork2014}(C.Dwork and A.Roth(2014)).
%\end{proof}
\section{Algorithms}
In this section, we put forward three privacy preserving algorithms for $l_1$ regression and calculate their privacy parameters respectively.
\subsection{Algorithm 1}
%In the regression setting,
%let $Y_i$ be the response for case $i$,  $x_{ij}$ be the value of predictor $j$ for case
%$i$, and $\beta_j$ be the regression coefficient corresponding to predictor $j$.
%Denote a data set of $n$ i.i.d. samples as $\bf{X}$ and each observation contains $d$ variables $x_1,x_2,\cdots,x_d$. For convenience, assume that
%$|Y_i| \leq B$ and $||x_i||_1 \leq 1$, $i=1,\cdots,n$.
%In this paper we consider the linear $l_1$ estimation problem, i.e.,
%we consider the problem of minimizing the functional
%\begin{equation}\label{E0}
%  F(\beta)=\frac{1}{n}\sum^n_{j=1}|r_j(\beta)|
%\end{equation}
%where $X_j$ means $j$-th row of {\bf{X}} and
%$$r_j(\beta)=\beta_0+{\bf{X_j}}\beta-Y_j,   j=1,2,\cdots,n $$,
%$$r(\beta)=\beta_0+{\bf X}\beta-Y$$
%is a set of linear functionals in $R^n$.Ridge penalized regression can be  viewed as minimizing the criterion
%\begin{equation}\label{Ridge1}
%  L(\beta) = F(\beta) + \frac{\lambda}{2} \beta'\beta
%\end{equation}

%finite smoothing method
The finite smoothing method is an important tool to solve nondifferentiable problem, for instance, median regression proposed in \cite{FM1993}.
In addition, \cite{FM1993} proves that the solution of smooth function can estimate the solution of original function well.
This idea is applied in algorithm $1$ by an analogous technique.

Since the absolute value function is not differentiable at the cuspidal point, a smooth method for minimizing function (\ref{Ridge1}) is considered.
%, which is a useful smooth method.
%\begin{equation}\label{Fsmooth}
%  F_\gamma = \frac{1}{n}\sum_{j=1}^n \rho_\gamma(r_j(\beta))
%\end{equation}
%where
Let $\gamma$ be a
nonnegative parameter which indicates the degree of approximation.
Define
%Because the absolute value function $|t|$ is nondifferentiable at zero, a smooth function $\rho_\gamma(t)$ is adopted to approximate it, where .
\begin{eqnarray}\label{rhoF}
\rho_\gamma(t)=    %方程组开始
\left\{                        %方程组的左边包括大括号\{
\begin{array}{rl}       %设定列阵的格式：{lll} 是三个L，表示三列的对齐方式为Left对齐
t^2/(2\gamma); & {\text{if}}~~ |t| \le \gamma, \\  %$――分隔列的标记，\\――表示换行
|t|-\frac{1}{2}\gamma; & { \text{if}}~~ |t|>\gamma.
\end{array}              %方程列阵的结束
\right.               %方程组的右边无符号，利用“.“来标示
\end{eqnarray}
Then the nondifferentiable function $F( \mu,{{\bm \beta}})$ is approximated by the Huber M-estimator (see \cite{HE1973}).

 Denote $F_\gamma( \mu,{\bm \beta}) = \frac{1}{n}\sum_{i=1}^n \rho_\gamma(r_i(\mu,{\bm \beta}))$ and $L_\gamma(\mu,{\bm \beta})=F_\gamma(\mu,{\bm \beta}) + \frac{\lambda}{2} {\bm \beta} ^\mathrm{T}{\bm \beta}$.
The sign vector $s_\gamma(\mu,{\bm \beta})= (s_1(\mu,{\bm \beta}),\ldots, s_n(\mu,{\bm \beta}))^\mathrm{T}$ is given by
\begin{eqnarray}\label{sign}
s_i(\mu,{\bm \beta})=    %方程组开始
\left\{                        %方程组的左边包括大括号\{
\begin{array}{rl}       %设定列阵的格式：{lll} 是三个L，表示三列的对齐方式为Left对齐
-1; & {\text{if}}~~~ r_i(\mu,{\bm \beta}) < -\gamma, \\  %$――分隔列的标记，\\――表示换行
0; &  {\text{if}}~~~  -\gamma \le r_i(\mu,{\bm \beta}) \le \gamma, \\
1; &  {\text{if}}~~~ r_i(\mu,{\bm \beta}) > \gamma.
\end{array}              %方程列阵的结束
\right.                    %方程组的右边无符号，利用“.“来标示
\end{eqnarray}
Let $w_i(\mu,{\bm \beta})=1-s_i^2( \mu,{\bm \beta})$, then
\begin{equation}\label{rhoquad}
\rho_\gamma(r_i(\mu,{\bm \beta})) =\frac{1}{2\gamma}w_i(\mu,{\bm \beta})r_i^2(\mu,{\bm \beta})+
s_i(\mu,{\bm \beta})\left[r_i(\mu,{\bm \beta})-\frac{1}{2}\gamma s_i(\mu,{\bm \beta}) \right ].
\end{equation}
Denote ${\bf{W}}_\gamma(\mu,{\bm \beta})$ as the diagonal $n \times n$ matrix whose diagonal elements are $w_i(\mu,{\bm \beta})$.
So ${\bf{W}}_\gamma(\mu,{\bm \beta})$ has value $1$ in those diagonal elements related to small residuals and $0$ elsewhere.
For $\mu \in \mathbb{R}$ and ${\bm \beta} \in {\mathbb{R}}^d$, the derivation of $F_\gamma(\mu,{\bm \beta})$ is
$$
 \frac{\partial F_\gamma(\mu,{\bm \beta})}{\partial {\bm \beta}} = \frac{1}{n}{\bf X}^\mathrm{T} \left[ \frac{1}{\gamma} {\bf W}_\gamma(\mu,{\bm \beta}) r(\mu,{\bm \beta})+s_\gamma(\mu,{\bm \beta})
 \right],
$$
and
$$
 \frac{\partial F_\gamma(\mu,{\bm \beta})}{\partial \mu} = \frac{1}{n}{\bf 1}^\mathrm{T} \left[ \frac{1}{\gamma} {\bf W}_\gamma(\mu,{\bm \beta}) r(\mu,{\bm \beta})+s_\gamma(\mu,{\bm \beta})
 \right].
$$
%and for $\beta \in {\mathbb{R}^d/ \bf B}_\gamma$ ( $\beta \in  \notin \mathbb{R}^d$ and $\beta \notin {\bf B}_\gamma$), where
%${\bf B}_\gamma= \{\beta \in {\mathbb{R}}^d|\exists j:|r_j(\beta)|=\gamma\}$,  the Hessian exists and is given by
%\begin{equation}\label{Hess}
%  F''_\gamma(\beta)=\frac{1}{n\gamma} {\bf X}^\mathrm{T}{\bf W}_\gamma(\beta)
%  {\bf X}
%\end{equation}
It can be verified that $L_\gamma(\mu , {\bm \beta})$ is convex and a minimizer  of $L(\mu,{\bm \beta})$ is close to a minimizer of $L_\gamma(\mu,{\bm \beta})$ when $\gamma$ is close to zero. Furthermore, according to Theorem $1$ in \cite{FM1993},
the $l_1$ solution can be detected when $\gamma > 0$ is small enough, i.e., it is not necessary
to let $\gamma$ converge to zero in order to find a minimizer of $L_\gamma(\mu,{\bm \beta})$. This observation is essential
for the efficiency and the numerical stability of the algorithm to be described in this
paper. In addition, refer to the algorithm in  \cite{LR2009}, the first privacy preserving algorithm for median regression is stated as follows.
\newline
{\bf{Algorithm 1:}} \newline
 \textbf{Inputs: privacy parameter $\epsilon$, design matrix
${\bf X}$, response vector ${\bf Y}$, regularization parameter $\lambda$ and approximation parameter $\gamma$ }.
    \newline
     \textbf{Generate a random vector $b$ from the density function $h(b) \propto exp{(-\frac{\epsilon}{4} ||b||_1)}$}. To implement this,
pick the $l_1$ norm of $b$ from the Gamma distribution $\Gamma(d+1, \frac{4}{\epsilon})$, and the direction of $b$ uniformly at
random.
    \newline
     \textbf{Compute $(\mu^*,{\bm \beta^*}) = {\text{argmin}}_{\mu, {\bm \beta}}
L_\gamma(\mu,{\bm \beta^*}) + \frac{b^\mathrm{T}{\bf \omega}}{n}$} + $\frac{\mu^2}{\sqrt{n}}$, where ${\bm \omega} = (\mu,{\bm \beta})$ is a $d+1$ dimensional vector, and $n$ is the number of rows of ${\bf X}$ .
    \newline
\textbf{Output $(\mu^*,{\bm \beta^*})$.}

This algorithm is  very similar to the smoothing median
regression convex program in \cite{FM1993}, and therefore its running time is similar to that of smoothing regression.
In fact, $(\mu^*,{\bm \beta^*})$ can be obtained by the interior point method.
Similar to the proof in \cite{LR2009}, we can show that Algorithm $1$ is privacy preserving.
\begin{theorem}\label{Th1}
Given a set of $n$ samples ${\bf X}_1, \ldots , {\bf X}_n$ over $\mathbb{R}^d$, with labels $Y_1, \ldots , Y_n$, where for each
$i$, $||{\bf X}_i||_1 \le 1$ and $||Y_i||_1 \le B$, the output of Algorithm 1 preserves $(\epsilon,0)$-differential privacy.
\end{theorem}
\begin{proof}
 %\rev{Assume $D'$ is a set of $n$ samples $X_1, \ldots , X_n$ over $\mathbb{R}^d$, with labels $Y_1, \ldots , Y_n$,
% where for each $i$, $||X_i||_1 \le 1$ and $|Y_i| \le B$. Fix $D=D' \cup \{(X_{n+1}, Y_{n+1})\}$, where $||X_{n+1}||_1 \le 1$ and
% $|Y_{n+1}| \le B$.}
\rev{Let ${\bm a}_1$ and ${\bm a}_2$ be two row vectors over $\mathbb{R}^d$ with $l_1$ norm at most $1$ and $y_1,y_2 \in [-B,B]$.
Consider the two inputs $D_1$ and
$D_2$ where $D_2$ is obtained from $D_1$ by replacing one record $({\bm a}_1,y_1)$ into $({\bm a}_2,y_2)$. For convenience, assume the first $n-1$ records are same.}
For any output ${\bm \omega^*} = (\mu^*,{\bm \beta^*})$ by Algorithm $1$, there is a unique
value of ${\bm b}$ that maps the input to the output. This uniqueness holds, because both the regularization
function and the loss functions are differentiable everywhere. Denote \rev{$\tilde{{\bm a}}_1$ as $(1, {\bm a}_1)$} and $\tilde{{\bm a}}_2$ as $(1, {\bm a}_2)$ .
Let the values of \rev{$d+1$ dimensional vector} ${\bm b}$ for $D_1$ and $D_2$ respectively, be ${\bm b}_1$ and ${\bm b}_2$.
Since ${\bm \omega^*}$
is the value that minimizes both the optimization problems, the derivative of both optimization functions at ${\bm \omega^*} $
is $0$.
This implies that for every ${\bm b}_1$ in the first case, there exists a ${\bm b}_2$ in the second case such that:
\begin{eqnarray*}
\begin{array}{lll}       %设定列阵的格式：{lll} 是三个L，表示三列的对齐方式为Left对齐
\quad & {\bm b}_1+\tilde{{\bm a}}_1^\mathrm{T}( 1/\gamma {\bf W}_\gamma( \mu^*, {\bm \beta^*})( \mu^*+{\bm a}_1^\mathrm{T}{\bm \beta^*}-y_1)+S_\gamma(\mu^*,{\bm \beta^*}))    \\  %$――分隔列的标记，\\―― 表示%换行
= & {\bm b}_2+\tilde{{\bm a}}_2^\mathrm{T}(1/\gamma {\bf W}_\gamma(\mu^*,{\bm \beta^*})(\mu^*+{\bm a}_2^\mathrm{T}{\bm \beta^*}-y_2)+S_\gamma(\mu^*,{\bm \beta^*})).
\end{array}              %方程列阵的结束
\end{eqnarray*}

According to the definitions of ${\bf W}_\gamma( \mu^*,{\bm \beta^*})$ and $S_\gamma( \mu^*,{\bf {\bm \beta^*}})$, it is clear that
$$-1 \le 1/\gamma {\bm W}_\gamma(\mu^*,{\bm \beta^*})*(\mu^*+{\bm a}_1^\mathrm{T}{\bm \beta^*}-y_1)+S_\gamma(\mu^*,{\bm \beta^*}) \le 1$$
and
$$-1 \le 1/\gamma {\bf W}_\gamma(\mu^*,{\bm \beta^*})*(\mu^*+{\bm a}_2^\mathrm{T}{\bm \beta^*}-y_2)+S_\gamma(\mu^*,{\bm \beta^*}) \le 1.$$
Since $||\tilde{{\bm a}}_1||_1 \le 2$ and $||\tilde{{\bm a}}_2||_1 \le 2$, we have $||{\bm b}_1-{\bm b}_2||_1 \le 4$, which implies that $-4 \le ||{\bm b}_1||_1-||{\bm b}_2||_1 \le 4$. Therefore, for any
$({\bm a}_1,y_1)$ and $({\bm a}_2,y_2)$,
$$
\frac{P((\mu^*,{\bm \beta^*})|{\bf X}_1, \ldots, {\bf X}_{n-1}, Y_1, \ldots, Y_{n-1}, {\bf X}_n={\bm a}_1,Y_n=y_1)}{P((\mu^*,{\bm \beta^*})|{\bf X}_1, \ldots, {\bf X}_{n-1}, Y_1, \ldots, Y_{n-1}, {\bf X}_n={\bm a}_2,Y_n=y_2)}
=\frac{h({\bm b}_1)}{h({\bm b}_2)}=e^{-\frac{\epsilon}{4}(||{\bm b}_1||_1-||{\bm b}_2||_1)},
$$
where $h({\bm b}_i)$ for $i=1,2$ is the density of ${\bm b}_i$. Since $-4 \le ||{\bm b}_1||_1-||{\bm b}_2||_1 \le 4$, this ratio is at most $exp(\epsilon)$.
%The differential privacy property of $\mu^*$ can be proved in a similar way.
\end{proof}
%\subsubsection*{Convergence}
 According to  Lemma $1$ in \cite{LR2009},  theoretical results for accuracy of parameter estimation is given for Algorithm $1$.
\begin{lemma}\label{le1}
Let $G({\bm \omega})$ and $g({\bm \omega})$ be two convex functions, which are continuous and differentiable at
all points. If ${\bm \omega}_1 = argmin_{\bm \omega} G({\bm \omega})$ and ${\bm \omega}_2 = argmin_{\bm \omega} G({\bm \omega}) + g({\bm \omega})$, then
$||{\bm \omega}_1 - {\bm \omega}_2||_1 \leq \frac{g_1}{G_2}$. Here,
$g_1 = \max_{\bm \omega} ||\triangledown g({\bm \omega})||_1$ and $G_2 = \min_{\bm v} \min_{\bm \omega} {\bm v}^\mathrm{T}\triangledown ^2G({\bm \omega}){\bm v}$, for any unit vector ${\bm v}$.
\end{lemma}
The main idea of the proof is to examine the gradient and the Hessian of the functions $G$ and $g$
around ${\bm \omega}_1$ and ${\bm \omega}_2$.
\begin{lemma}\label{LapLe}
If $||{\bm b}||_1$ is a random variable drawn from
$\Gamma(d+1,\frac{4}{\epsilon})$, then with possibility $1-\alpha$, $||{\bm b}||_1 \leq \frac{4(d+1) log(\frac{d+1}{\alpha})}{\epsilon}$.
\end{lemma}
\begin{proof}
 Since a random variable drawn from $\Gamma(d+1,\frac{4}{\epsilon})$ can be written as the sum of $d+1$ independent identically distributed random variables, each of which is distributed as an exponential random variable with mean $\frac{4}{\epsilon}$. Using an union bound, we see that with probability $1-\alpha$, the values of all $d+1$ of these variables are upper bounded by $\frac{4log(\frac{d+1}{\alpha})}{\epsilon}$. Therefore, with probability at least $1-\alpha$,
 $||{\bm b}||_1 \le  \frac{4(d+1) log(\frac{d+1}{\alpha})}{\epsilon}$.
\end{proof}
\begin{theorem}\label{le2}
Given an $l_1$ regression problem with regularization parameter $\lambda$, let ${\bm \omega}_1$ be the classifier
that minimizes $L_\gamma(\mu,{\bm \beta})+ \frac{\mu^2}{\sqrt{n}}$, and ${\bm \omega}_2$ be the classifier output by Algorithm $1$ respectively. Then, with
probability $1-\alpha$, $||{\bm \omega}_1-{\bm \omega}_2||_1 \leq \frac{4(d+1) log(\frac{d+1}{\alpha})}{n \min(\lambda, \frac{2}{\sqrt{n}}) \epsilon}$.
\end{theorem}
\begin{proof}
\rev{According to Lemma \ref{le1}, we take $G({\bm \omega}) = L_\gamma(\mu,{\bm \beta}) + \frac{\mu^2}{\sqrt{n}}$  and $g({\bm \omega}) = \frac{b^\mathrm{T}{\bm \beta}}{n}$.
Because
$F_\gamma(\mu,{\bm \beta})$ is a convex function, if we define the second derivative of $F_\gamma(\mu,{\bm \beta})$ is $0$ at nondifferentiable points, then the hessian matrix of $F_\gamma(\mu,{\bm \beta})$ is positive semidefinite. Notice that $\triangledown^2(\frac{\mu^2}{\sqrt{n}}) = \frac{2}{\sqrt{n}}$ and $\triangledown^2 (\frac{\lambda}{2}\beta^\mathrm{T}{\bm \beta}) =\lambda {\bf I}$, where $\bf I$ is an identity matrix with size $d \times d$.
Hence, for any unit vector $v$,
$G_2 =  \min_{\bm v} \min_{\bm \omega} {\bm v}^\mathrm{T}\triangledown^2G({\bm \omega}){\bm v} \ge \min(\lambda, \frac{2}{\sqrt{n}})$ and $g_1 = \frac{||{\bm b}||_1}{n}$,  $||{\bm \omega}_1-{\bm \omega}_2||_1 \leq \frac{||{\bm b}||_1}{n\min(\lambda, \frac{2}{\sqrt{n}})}$. Since $b$ is a random variable drawn from
$\Gamma(d+1,\frac{4}{\epsilon})$, according to Lemma \ref{LapLe}, with possibility $1-\alpha$, $||{\bm b}||_1 \leq \frac{4(d+1) log(\frac{d+1}{\alpha})}{\epsilon}$, then the theorem is obtained.}
%Because
%$G_2 \geq \lambda$ and $g_1 \leq \frac{b}{n}$,  $||w_1-w_2||_1 \leq \frac{b}{n\lambda}$. Since $b$ is a random variable drawn from
%$\Gamma(d,\frac{2}{\epsilon})$, with possibility $1-\alpha$, $||b||_1 \leq \frac{2d log(\frac{d}{\alpha})}{\epsilon}$, the theorem is obtained.
\end{proof}
When $n$ is sufficient large, ${\bm \omega}_2$ approximates ${\bm \omega}_1$ well and ${\bm \omega}_1$ is close to true parameter of $argmin_{\bm \omega}$ $L_\gamma({\bm \omega})$.
\subsection{Algorithm 2}%iterative method 1973
The second algorithm is based on the iterative algorithm, which was first proposed in \cite{S1973}.
This iterative technique combines  absolute deviations regression with least square regression. Hence, at the heart of the technique is any standard least squares curve fitting algorithm.

The basic least squares algorithm minimizes the
criterion
\begin{equation}\label{Lag1}
I = \frac{1}{n}\sum_{i=1}^n w_i r^2_i(\mu,{\bm \beta}) +\frac{\lambda}{2}{\bm \beta}^\mathrm{T}{\bm \beta},
\end{equation}
where the weighting factors $w_i$ are positive real numbers.
Based on the Lagrange multiplier approach, for a fixed $\lambda$, there exists a unique value $v$ such that  minimizing equation (\ref{Lag1}) is equivalent to  minimizing the following equation.
\begin{eqnarray*}\label{Lag2}
  I = \frac{1}{n}\sum_{i=1}^n w_i r^2_i(\mu,{\bm \beta}), \\
  s.t. \quad {\bm \beta}^\mathrm{T} {\bm \beta} \leq v.
\end{eqnarray*}
 Considering the $(t+1)$-th iteration, we take $w_i$ as $\frac{1}{|r(t)_i|+e}$, where $r(t)_i$ is the residual of $i$-th sample at the $t$-th iteration. Then the iterative process can be written as

\begin{equation}\label{Iter1973}
 I(t+1) = \frac{1}{n}\sum_{i=1}^n \frac{1}{|{r(t)}_i|+e}{r^2(t+1)}_i+\frac{\lambda}{2}{\bm \beta}^\mathrm{T}{\bm \beta}.
\end{equation}
If $||r(t)_i-r(t+1)_i||_1\approx 0,i=1,2,\ldots,n$, (\ref{Iter1973}) is close to $L(\mu,{\bm \beta})$. In practice, we set $e$ as a small positive value (such as $e =0.05$) .
\newline
{\bf{Algorithm 2:}} \newline
 \textbf{Inputs: privacy parameter $\epsilon$, deign matrix
${\bf X}$, response vector ${\bf Y}$, regularization parameter $\lambda$, tolerance parameter $\tau$ and the number of iteration $N_0$ \newline
 Initialize the algorithm with  $\hat{\mu}(0)$ and $\hat{{\bm \beta}}(0)$
    \newline
    \hspace*{2cm} $( \hat{\mu}(1), \hat{{\bm \beta}}(1)) = \text{argmin}_{\mu, {\bm \beta}} I(1)$
    \newline
    \hspace*{1cm} for $t=1,\cdots, N_0-1$ do
    \newline
    \hspace*{2cm}
    while $||\hat{\mu}(t)-\hat{\mu}(t-1)||_1 > \tau$ or $||\hat{{\bm \beta}}(t)-\hat{{\bm \beta}}(t-1)||_1 > \tau$ do
    \newline
    \hspace*{2cm} $( \hat{\mu}(t+1), \hat{{\bm \beta}}(t+1)) = \text{argmin}_{\mu, {\bm \beta}} I(t+1)$
    \newline
    \hspace*{2cm}else do
    \newline
    \hspace*{2cm}Output $( \hat{\mu}(N_0), \hat{{\bm \beta}}(N_0)) : =(\hat{\mu}(t), \hat{{\bm \beta}}(t))$
    \newline
    \hspace*{2cm}break
    \newline
    \hspace*{2cm}end while
    \newline
    \hspace*{1cm}end for
    \newline
    Output $( \hat{\mu}, \hat{{\bm \beta}}) : =(\hat{\mu}(N_0), \hat{{\bm \beta}}(N_0)) + {\bf{U}}$,
    \newline
     where ${\bf{U}}$ is a $d+1$ dimensional
    Laplace random variable with parameter
    \newline
    $c =\frac{8}{n \min(\frac{2}{2(\sqrt{dv}+B)+e},\lambda) e}(\sqrt{d v}+B) $
        \newline
}

\begin{theorem}
Given a set of $n$ samples ${\bf X}_1, \ldots , {\bf X}_n$ over {$\mathbb{R}^d$}, with labels $Y_1, \ldots , Y_n$, where for each $i$
($1 \le i \le n$), $||{\bf X}_i||_1 \le 1, |Y_i| \le B$, the output of Algorithm 2 preserves $(\epsilon,0)$-differential privacy.
\end{theorem}
\begin{proof}
 Denote ${\bm \omega} = (\hat{\mu}(N_0),\hat{{\bm \beta}}(N_0))$ and the $l_1$ sensitivity of  ${\bm \omega}$ as $s({\bm \omega})$.
 \rev{Let ${\bm a}_1$ and ${\bm a}_2$ be two vectors over $\mathbb{R}^d$ with $l_1$ norm at most $1$ and $y_1,y_2 \in [-B,B]$.
Consider the two inputs $D_1$ and
$D_2$ where $D_2$ is obtained from $D_1$ by changing one record $({\bm a}_1,y_1)$ into $({\bm a}_2,y_2)$. For convenience, assume the first $n-1$ records are same.}
 %\rev{ Let $D'$ be a set of $n$ samples $X_1, \ldots , X_n$ over {$\mathbb{R}^d$}, with labels $Y_1, \ldots , Y_n$, where for each
%$i$, $||X_i||_1 \le 1, |Y_i| \le B$. Fix $D=D' \cup \{(X_{n+1}, Y_{n+1})\}$, where $||X_{n+1}||_1 \le 1$ and
% $|Y_{n+1}| \le B$.}
 %Let $a_1$ and $a_2$ be two vectors over {$\mathbb{R}^d$} with $1$-norm at most 1 and $y_1,y_2 \in {\mathbb{R}}$. For any such $(a_1, y_1),(a_2, y_2)$, consider the inputs $(X_1, Y_1), \ldots ,(X_{n-1}, Y_{n-1}),(a_1, y_1)$ and
%$(X_1, Y_1),\ldots,(X_{n-1}, Y_{n-1}),(a_2, y_2)$.
According to Lemma \ref{le1}, let
$G({\bm \omega}) =I(N_0)$ and \rev{$g({\bm \omega}) = \frac{1}{n}w_2(\hat{\mu}(N_0)+a_2^\mathrm{T}\hat{{\bm \beta}}(N_0)-y_2)^2- \frac{1}{n}w_1(\hat{\mu}(N_0)+{\bm a}_1^\mathrm{T}\hat{{\bm \beta}}(N_0)-y_1)^2$}. Similar to the proof in Theorem \ref{le2}, we can achieve that
$$g_1 = max_{\bm \omega} ||\triangledown g({\bm \omega})||_1 \le \frac{2}{n}|w_1|(|\hat{\mu}(N_0)|+|{\bm a}_1^\mathrm{T}\hat{{\bm \beta}}(N_0)|+|y_1|)+\frac{2}{n}|w_2|(|\hat{\mu}(N_0)|+|{\bm a}_2^\mathrm{T}\hat{{\bm \beta}}(N_0)|+|y_2|).
$$
Notice that $( \hat{\mu}(N_0), \hat{{\bm \beta}}(N_0)) = \text{argmin}_{\mu, {\bm \beta}} I(N_0)$, then $\frac{\partial I(N_0)}{\partial \mu} = 0$ at $\mu = \hat{\mu}(N_0)$, that is,
\begin{eqnarray*}
\begin{array}{rrr}       %设定列阵的格式：{lll} 是三个L，表示三列的对齐方式为Left对齐
\quad &\sum_{i=1}^{n}w_i(\hat{\mu}(N_0)+{\bf X}_i \hat{\beta}(N_0)-Y_i) = 0    \\  %$――分隔列的标记，\\――表示%换行
\Longleftrightarrow & \hat{\mu}(N_0) = \frac{-\sum_{i=1}^{n}w_i({\bf X}_i\hat{\beta}(N_0)-Y_i)}{\sum_{i=1}^{n}w_i}.
\end{array}              %方程列阵的结束
\end{eqnarray*}
  Since $0< w_i \leq 1/e$ ,  $||y_i||_1 \leq B( i=1,\cdots,n)$ and
 $||\hat{{\bm \beta}}(N_0)||_1 \leq \sqrt{d} ||\hat{{\bm \beta}}(N_0)||_2 \leq \sqrt{dv}$,
 we have $||\hat{\mu}(N_0)||_1 \le \sqrt{d v}+B$.  Notice that above inequalities are still true in $t$-th($ \ge 2$) iteration and hence $\frac{1}{2(\sqrt{dv}+B)+e}\le w_i \le \frac{1}{e}.$    %\rev{$g_1 \leq \frac{4}{ne}(\sqrt{d v}+B)$.
 Then we can achieve that
$$g_1 = max_{\bm \omega} ||\triangledown g({\bm \omega})||_1 \le \frac{8(\sqrt{dv}+B)}{ne}.
$$
 \rev{In addition, denote $F_e(\omega) =\frac{1}{n}\sum_{i=1}^{n}w_ir^2_i(\mu,\beta)$. It can be checked that $F_e(\omega)$ is convex and
 $\frac{\partial F_e^2(\omega)}{\partial \mu^2}=\frac{2}{n}\sum_{i=1}^nw_i \le \frac{2}{2(\sqrt{dv}+B)+e}$, $\triangledown^2 (\frac{\lambda}{2}\beta^\mathrm{T}\beta) =\lambda {\bf I}$, where $\bf I$ is an identity matrix with size $d \times d$,
 then $G_2 \ge \min(\frac{2}{2(\sqrt{dv}+B)+e}, \lambda)$ and
 $s({\bm \omega}) \le \frac{8}{n \min(\frac{2}{2(\sqrt{dv}+B)+e},\lambda) e}(\sqrt{d v}+B)$ }.
%Notice that $( \hat{\mu}(N_0), \hat{\beta}(N_0)) = \text{argmin}_{\mu, \beta} I(N_0)$, then $\frac{\partial I(N_0)}{\partial \mu} = 0$ at $\mu = \hat{\mu}(N_0)$, that is,
%\begin{eqnarray*}
%\begin{array}{rrr}       %设定列阵的格式：{lll} 是三个L，表示三列的对齐方式为Left对齐
%\quad &\sum_{i=1}^{n}w_i(\hat{\mu}(N_0)+X_i \hat{\beta}(N_0)-Y_i) = 0    \\  %$――分隔列的标记，\\――表示%换行
%\Longleftrightarrow & \hat{\mu}(N_0) = \frac{-\sum_{i=1}^{n}w_i(X_i\hat{\beta}(N_0)-Y_i)}{\sum_{i=1}^{n}w_i}.
%\end{array}              %方程列阵的结束
%\end{eqnarray*}
%  In addition, since $0< w_i \leq 1/e$ ,  $y_i \leq B( i=1,\cdots,n)$ and
% $||\hat{\beta}(N_0) \leq \sqrt{d} ||\hat{\beta}(N_0)||_2 \leq \sqrt{dv}$,
% we have $||\hat{\mu}(N_0)||_1 \le \sqrt{d v}+B$ and  \rev{$g_1 \leq \frac{4}{ne}(\sqrt{d v}+B)$, then $s(N_0) \le \frac{4}{n\lambda e}(\sqrt{d v}+B)$ }.
%\begin{eqnarray*}
%\begin{array}{lll}       %设定列阵的格式：{lll} 是三个L，表示三列的对齐方式为Left对齐
%|s(t+1)| & \leq \frac{1}{n}|w_1|(|\mu|+|a_1^\mathrm{T}\beta|+|y_1|)+\frac{1}{n}|w_2|(|\mu|+|a_2^\mathrm{T}\beta|+|y_2|) \\  %$――分隔列的标记，\\――表示换行
%|s(t+1)| & \leq \frac{2}{ne}(\sqrt{d v}+2B).
%\end{array}              %方程列阵的结束
%\end{eqnarray*}
%\begin{eqnarray}{ll}
%% \nonumber to remove numbering (before each equation)
% |s(t+1)| \leq \frac{1}{n}|w_1|(|\mu|+|a_1^\mathrm{T}\beta|+|y_1|)+\frac{1}{n}|w_2|(|\mu|+|a_2^\mathrm{T}\beta|+|y_2|) \\
%   |s(t+1)|\leq \frac{2}{ne}(v+B).
%\end{eqnarray}

 According to lemma \ref{Lap}, the result is obtained directly from the composition theorem.
\end{proof}
%\newline
% \textbf{Let $s$ = $s_0$ \newline
% For $k$ = 0 through $kmax$:\newline
% \hspace*{1cm} $T$ $\leftarrow$ temperature($k$)\newline
% \hspace*{1cm}    Pick a random neighbour: $s'$ $\leftarrow$ neighbour($s$)\newline
% \hspace*{1cm}If E($s'$)$<$E($s$) or P(E($s$), E($s'$), $T$)$ <$ random$(0, 1)$:\newline
%    \hspace*{2cm}$s$ $\leftarrow$ $s'$\newline
% \hspace*{1cm}EndIf \newline
% EndFor \newline
% Output: the final state $s$\newline}

%\subsubsection*{Convergence}

  For $e >0$, define  a perturbation of $L(\mu,{\bm \beta})$ as
  $$L_e(\mu, {\bm \beta})= \sum_{i=1}^{n}|r_i(\mu,{\bm \beta})| - \frac{e}{2}ln(e+|r_i(\mu,{\bm \beta})|) + \frac{\lambda}{2}{\bm \beta}^{\mathrm{T}}{\bm \beta}.$$
  \cite{DR2000} proves that iterative least square algorithm without adding noise is a special case of Majorization-Minimization ( MM) algorithms ( see \cite{DR2004}) for objective function $L_e(\mu,{\bm \beta})$ and obtained convergence results.
 \begin{proposition}\label{covg2}
 For linear median regression with a full-rank covariate matrix {\bf X},
the iterative least square algorithm without adding noise converges to the unique minimizer of $L_e(\mu,{\bm \beta})$.
 \end{proposition}
 \begin{proposition}
 If ($\hat{\mu}_e, \hat{{\bm \beta}}_e)$ minimizes $L_e(\mu,{\bm \beta})$, then any limit point of ($\hat{\mu}_e, \hat{{\bm \beta}}_e)$ as $e$ tends to $0$
minimizes $L(\mu,{\bm \beta})$. If $L(\mu, {\bm \beta})$ has a unique minimizer $(\tilde{\mu}, \tilde{{\bm \beta}})$, then
$lim_{e \rightarrow 0}(\hat{\mu}_e, \hat{{\bm \beta}}_e)$ = $(\tilde{\mu}, \tilde{{\bm \beta}})$.
 \end{proposition}
The proof of above propositions can be seen in \cite{DR2000}.
 \begin{theorem}\label{con2}
Given a $l_1$ regression problem with regularization parameter $\lambda$, let ${\bm \omega}_1$ be the classifier
that minimizes $L_e(\mu,{\bm \beta})$, and ${\bm \omega}_2$ be the classifier output by Algorithm $2$ respectively. Then, with
probability $1-\alpha$, $||{\bm \omega}_1-{\bm \omega}_2||_1 \leq \rev{\frac{8(\sqrt{d v}+B)(d+1) log(\frac{d+1}{\alpha})}{\epsilon \min(\frac{2}{2(\sqrt{dv}+B)+e},\lambda) ne}}$.
\end{theorem}
\begin{proof}
Since $||{\bm b}||_1$ is a random variable drawn from
$\Gamma(d+1,\frac{8(\sqrt{d v}+B)}{\epsilon \min(\frac{2}{2(\sqrt{dv}+B)+e},\lambda) ne})$, with possibility $1-\alpha$, $||{\bm b}||_1 \leq \frac{8(\sqrt{d v}+B)(d+1) log(\frac{d+1}{\alpha})}{\epsilon \min(\frac{2}{2(\sqrt{dv}+B)+e},\lambda) ne}$, the theorem is obtained.
\end{proof}

Therefore, for fixed small $e$, if $n$ is sufficient large, accuracy can be ensured in practice.
 %\begin{corollary}
% Given a set of $n$ samples $X_1, \ldots , X_n$ over $\mathbb{R}^d$ (for each
%$i$, $||X_i|| \le 1$) with labels $Y_1, \ldots , Y_n$(for each
%$i$, $||Y_i|| \le B$), regularization parameter $\lambda$( and $v$), fixed positive parameter $e$,  and the number of iterations $N_0$, Algorithm $2$ is convergent in probability with rate $O(\frac{1}{n})$.
% \end{corollary}

\subsection{Algorithm 3}%gradient method
In \cite{DL2016}, the authors argue that
adding noise to the estimated parameters after optimization would destroy the utility of the learned model.
Hence, we prefer a more sophisticated method to control the influence of the training data during
the training process, especially in the stochastic gradient decent computation.
\cite{Wu2008} declares that greedy coordinate descent is an effective method for $l_1$ regression,
where $l_1$ regression means median regression. So we apply this idea to minimize objective function $L(\mu,{\bm \beta})$ in a similar way. Although $L(\mu,{\bm \beta}) $ is nondifferentiable, it does possess directional derivatives along each forward or backward coordinate direction. For example, if ${\bm e}_k$ is the coordinate direction along which $\beta_k$ varies, then the objective function (2) has directional derivatives
$$
d_{{\bm e}_k^+}L(\mu,{\bm \beta}) = {lim}_{\tau\rightarrow 0^+}
\frac{L(\mu,{\bm \beta}+\tau {\bm e}_k)-L(\mu,{\bm \beta})}{\tau}=d_{{\bm e}_k^+}F(\mu,{\bm \beta})+\lambda \beta_k
$$
and
$$
d_{{\bm e}_k^-}L(\mu,{\bm \beta}) = {lim}_{\tau\rightarrow 0^-}
\frac{L(\mu,{\bm \beta}+\tau {\bm e}_k)-L(\mu,{\bm \beta})}{\tau}=d_{{\bm e}_k^-}F(\mu,{\bm \beta})+\lambda \beta_k.
$$
In $l_1$ regression,  the coordinate direction derivatives are
\begin{eqnarray}\label{l1der1}
d_{{\bm e}_k^+}F(\mu,{\bm \beta})= \frac{1}{n}\sum_{i=1}^n   %方程组开始
\left\{                        %方程组的左边包括大括号\{
\begin{array}{rll}       %设定列阵的格式：{lll} 是三个L，表示三列的对齐方式为Left对齐
-x_{ik}, & r_i(\mu,{\bm \beta})<0, \\  %$――分隔列的标记，\\――表示换行
x_{ik},  &  r_i(\mu,{\bm \beta})>0, \\
|x_{ik}|, &  r_i(\mu,{\bm \beta})=0,
\end{array}              %方程列阵的结束
\right.                      %方程组的右边无符号，利用“.“来标示
\end{eqnarray}
and
\begin{eqnarray}\label{l1der2}
d_{{\bm e}_k^-}F(\mu,{\bm \beta})= \frac{1}{n}\sum_{i=1}^n   %方程组开始
\left\{                        %方程组的左边包括大括号\{
\begin{array}{rll}       %设定列阵的格式：{lll} 是三个L，表示三列的对齐方式为Left对齐
x_{ik}, & r_i(\mu,{\bm \beta})<0, \\  %$――分隔列的标记，\\――表示换行
-x_{ik},&  r_i(\mu,{\bm \beta})>0, \\
|x_{ik}|,&  r_i(\mu,{\bm \beta})=0.
\end{array}              %方程列阵的结束
\right.              %方程组的右边无符号，利用“.“来标示
\end{eqnarray}
In greedy coordinate
descent progress\cite{JH2007}, we update the direction of parameter $\beta_k$ based on
min$\{d_{{\bm e}_k^+}L(\mu,{\bm \beta}), d_{{\bm e}_k^-}L(\mu,{\bm \beta})\}$. If both  coordinate directional derivatives are
nonnegative, the update of $\beta_k$ stops. In addition, $\hat{\mu} = \frac{1}{n_0} \sum^{n_0}_{i=1}(Y_i-{\bf X}_i\hat{{\bm \beta}})$, where $n_0 = n/N_0$. And by the method of batch
gradient \cite{W2003}, the $t$-th iteration only employs records with batch size $n_0$, which means $L(\hat{\mu}(t),\hat{{\bm \beta}}(t))$ in the algorithm is calculated by subset
({\bf X}(t), {\bf Y}(t)).  The algorithm is described as follows.
\newline
{\bf{Algorithm 3:}} \newline
 \textbf{Inputs: privacy parameters $\epsilon$, deign matrix
${\bf X}$, response vector ${\bf Y}$, regularization parameter $\lambda$, positive number $\ell$ and the number of iterations $N_0$.
 \newline
    Randomly split ( {\bf X}, {\bf Y}) into $N_0$ disjoint subsets
%    $({\bf X}(0), {\bf Y}(0)), ({\bf X}(1), {\bf Y}(1)), \cdots , ({\bf X}(N_0-1), {\bf Y}(N_0-1))$
    of size $n_0$.
    \newline
    Initialize the algorithm with a vector $(\hat{\mu}(0),\hat{{\bm \beta}}(0)$ (\textrm{ such as the solution of $l_2$ regression}).
    \newline
    for $t= 0, 1, 2, ..., N_0-1$ do
    \newline
    \hspace*{1cm} $\eta_t = \frac{\ell}{t+1}$
    \newline
    \hspace*{1cm} for $k=1,2,\cdots,d$ do
    \newline
    \hspace*{1cm}\hspace*{1cm}  $\hat{\beta}_k(t+0.5) = \hat{\beta}_k(t)-
        \eta_t (\text{min}~ \{d_{{\bm e}_k^+}L(\hat{\mu}(t),\hat{{\bm \beta}}(t)), d_{{\bm e}_k^-}L(\hat{\mu}(t),\hat{{\bm \beta}}(t))\})$,
        \newline
    \hspace*{1cm}\hspace*{1cm} $\hat{\beta}_k(t+1) = \hat{\beta}_k(t+0.5)+ U_t$,
        where $U_t \sim Lap(\frac{2\eta_t}{\epsilon n_0}), n_0=n/N_0$.
        \newline
    \hspace*{1cm} end for
    \newline
    \hspace*{1cm} $\hat{\mu}(t+1) = \frac{1}{n_0} \sum^{n_0}_{i=1} (Y_i-{\bf X}_i\hat{{\bm \beta}}(t+1))$.
    \newline
    end for
    \newline
    Output $\hat{\beta}:=\hat{\beta}(N_0)$, $\hat{\mu} = \hat{\mu}(N_0)$.
}

\begin{theorem}
Given a set of $n$ samples ${\bf X}_1, \ldots , {\bf X}_n$ over $\mathbb{R}^d$  with labels $Y_1, \ldots , Y_n$,  where for each
$i$($1\le i \le n$), $||{\bf X}_i||_1 \le 1$ and $||Y_i||_1 \le B$, the output of Algorithm $3$ preserves $(\epsilon,0)$-differential privacy.
\end{theorem}
\begin{proof}
Because of sample splitting, for $(x, y)  \in ({\bf X}(t), {\bf Y}(t))$
for
some $0 \le t  \le N_0-1$, it suffices to prove the privacy guarantee for the
$t$-th iteration of the algorithm: any iteration prior to the $t$-th one does not depend on $(x, y)$, while any iteration after the $t$-th one is differentially private by
post-processing \cite{Dwork2014}.

At the $t$-th iteration, the algorithm first updates the non-sparse estimate of $\beta_k$:
$$
\hat{\beta}_k(t+0.5) = \hat{\beta}_k(t)- \eta_t (\; \text{min}~ {d_{{\bm e}_k^+}L(\hat{\mu}(t),\hat{{\bm \beta}}(t)), d_{{\bm e}_k^-}L(\hat{\mu}(t),\hat{{\bm \beta}}(t)})).
$$
  \rev{Let $a_1$ and $a_2$ be two vectors over $\mathbb{R}^d$ with $l_1$ norm at most $1$ and $y_1,y_2 \in [-B,B]$.
Consider the two inputs $D_1$ and
$D_2$ where $D_2$ is obtained from $D_1$ by changing one record $(\;{\bm a}_1,y_1)$ into $(\;{\bm a}_2,y_2)$. For convenience, assume the first $n_0-1$ records are same.}
Denote $Dir_1(t)$ as the direction derivation
( $\text{min}~ \{d_{{\bm e}_k^+}L(\hat{\mu}(t),\hat{{\bm \beta}}(t)), d_{{\bm e}_k^-}L(\hat{\mu}(t),\hat{{\bm \beta}}(t))\})$ for the dataset $D_1$ and $Dir_2(t)$ for the dataset $D_2$.
Notice that $\hat{\beta}(t)$ does not depend on ( ${\bf X}(t),{\bf Y}(t))$, so $\hat{{\bm \beta}}(t+1)$ would be $(\epsilon,0)$-differentially private if it can be shown
that: for every pair $D$ and $D'$, we have \rev{
$$
|| \eta_t/n_0 \left [Dir_1(t)-Dir_2(t)\right ]||_1
%\le \eta_t/n_0(||X_{n_0+1}||_1)
\le \frac{2\eta_t}{n_0}.
$$
}
This is true, since $|| \eta_t/n_0 \left [Dir_1(t)-Dir_2(t)\right ]||_1
\le \eta_t/n_0(||{\bm a}_1||_1+||{\bm a}_2||_1)
\le \frac{2\eta_t}{n_0}$,
%where $x_k$ and $\tilde{x}_k$ represent $k$-th subscript of  vector $x$ and $\tilde{x}$ respectively.
then the privacy guarantee for ${\bm \beta}$ is proved by Lemma \ref{Lap}. In addition, since $\hat{\mu} = \frac{1}{n_0} \sum^{n_0}_{i=1}(Y_i-{\bf X}_i\hat{{\bm \beta}})$, it is differentially private by
post-processing \cite{Dwork2014}. Then the theorem is obtained.
\end{proof}

%\subsubsection*{Convergence}
\cite{Wu2008} said that coordinate descent may fail for a nondifferentiable function since all directional derivatives
must be nonnegative at a minimum point. However, if we can obtain a suitable approximate value quickly, this shortcoming can be accepted in practice. The following theorem shows that
estimated parameters would be stable when the number of iteration $N_0$ is large.
%When the number of iterations is sufficient large, the following theorem shows that estimated parameters is weakly consistent\cite{Wu2008}.
%\rev{
%\begin{description}
%  \item[S1] The parameter vector $\omega$ is confined to a compact domain $K \subset R^{d+1}$.
%The true parameter vector $\tilde{\omega}$ is an interior point of $K$.
%  \item[S2] The random errors $e_i =y_i-\mu-X_i\beta$ are independent; $e_i$ has distribution function $F_i(e)$
%      with $F_i(0)=1/2$
%  \item[S3] For every $c>0$, there exists an $f>0$ with
%  $$
%  \mathop{inf} \limits_i min\{F_i(c)-1/2,1/2-F_i(-c)\} \ge f
%  $$
%  \item[S4] The predictor vectors $Z_i=(1,X_i)$ satisfy $||z_i||_2 \le B$ for some $B \ge 0$.
%  \item[S5]For some $e>0$ and $d>0$, the predictors $z_i$ satisfy
%      $$
%      \mathop{inf}\limits_{||v||=1}\frac{1}{n}\sum_{i=1}^n I\{|z_iv| \ge e|\} \ge d
%      $$
%      for $n$ sufficiently large.
%\end{description}
%\begin{theorem}
%Under the regularity conditions $S1$ through $S5$, the sequence of
%estimators minimizing the criterion,  $L( \mu,\beta) =  \frac{1}{n}\sum^n_{i=1}|r_i( \mu,\beta)| + \frac{\lambda}{2} \beta^\mathrm{T} \beta$ is weakly
%consistent.
%\end{theorem}
%\begin{proof}
%\cite{Wu2008} proves that it is true for lasso penalized $l_1$ regression, the proof for ridge penalized $l_1$ regression is same.
%\end{proof}
%\begin{corollary}
%When $n$ is sufficient large, the sequence of
%estimators from Algorithm $3$ is weakly consistent.
%\end{corollary}}
\begin{theorem}
Given a set of $n$ samples ${\bf X}_1, \ldots , {\bf X}_n$ over $\mathbb{R}^d$   with labels $Y_1, \ldots , Y_n$( for each
$i$, $||{\bf X}_i|| \le 1$ and $||Y_i||_1 \le B$), Algorithm $3$ is convergent in probability with rate $O(1/t)$.
\end{theorem}
\begin{proof}
Consider the $t$-th iteration for $\beta_k$, since $|Dir(t)| \le \frac{1}{n_0}\sum_{i=1}^{n_0} ||x_{ik}||_1 \le 1$ and
$\hat{\beta}_k(t+1) = \hat{\beta}_k(t) - \eta_t Dir(t) + U_t$, $|\hat{\beta}_k(t+1)-\hat{\beta}_k(t)| \le |\eta_t|+ |U_t| = O_p(1/t)$, where $O_p(1/t)$ indicates that it converges in probability with rate $O(1/t)$. Since $\hat{\mu} = \frac{1}{n_0} \sum^{n_0}_{i=1} (Y_i-{\bf X}_i\hat{{\bm \beta}})$, it is
convergent in probability with rate $O(1/t)$, too. Then the theorem is obtained.
\end{proof}
\section{Simulated results}
Denote $n$ as the number of samples. Here let $n$ take two values: $5000$ and $5000000$. In fact, when $n$ is  small, such as $100$, Algorithm $1$ can perform well, but Algorithm $2$ requires $n$ bigger ( otherwise the noise added would be big which would result in big estimation error). Consider the following example with three independent variables $x_1, x_2, x_3$, where $y_i = 2 +3x_{i1}-4x_{i3} + u_i$ and $u_i$ obeys the
Laplace distribution $Lap(2)$, for $i=1, \ldots, n$. We assume that, for each $i$( $1\le i\le n$), $l_1$ norm of ${\bf X}_i$ is less than $1$ and $l_1$ norm of $Y_i$ is less than $2$. In practice, we take $\lambda$ as $0.002$ in the objective function. In Algorithm $1$, parameter $\gamma$ is taken as $0.05$. In Algorithm $2$, we set parameter $e =0.2$, tolerance parameter $\tau=10^{-6}$ and
the number of iteration $N_0 =  200$. In fact, Algorithm $2$ tends to converge with less than $30$ iterations. In Algorithm $3$, we set $\ell = 0.1$, step size $\eta_t=\frac{\ell}{t+1}$ and the number of iteration $N_0 = 40$. In addition, privacy parameters $(\epsilon,\delta)$=
$(0.1,0)$ for all the above algorithms. The results are listed in Table \ref{SimulatedResults1} and Table \ref{SimulatedResults2}. It shows that Algorithm $1$ performs better than the others when $n=5000$. However, when $n$ becomes much bigger, Algorithm $1$ costs much more time. Notice that when $n=5000000$, the noise added to Algorithm $2$ becomes small and it makes the estimated result precise. In addition, Algorithm $3$ costs less time in both cases, but it highly depends on initial value and step size $\eta_t$, which is a common problem for the gradient descent method \cite{W2003}.
\begin{table}[htbp]
  \centering
  \caption{Estimated results with sample size $5000$}
    \begin{tabular}{ccccc}
    \toprule
          & \multicolumn{1}{l}{Algorithm 1} & \multicolumn{1}{l}{Algorithm 2} & \multicolumn{1}{l}{Algorithm 3} & \multicolumn{1}{l}{True value} \\
    \midrule
    \textbf{$\mu$} & 2.0684 & 1.9440 & 1.8204 & 2 \\
    \textbf{$\beta_1$} & 3.0007 & 0.9227 & 2.6914 & 3 \\
    \textbf{$\beta_2$} & -0.0295 & 13.2762 & -0.6099 & 0 \\
    \textbf{$\beta_3$} & -4.0835 & -14.2089 & -3.2283 & -4 \\

    \textbf{time(s)} & 0.7113 & 0.3143 & 0.1220 &  \\
    \bottomrule
    \end{tabular}%
  \label{SimulatedResults1}%
\end{table}%

\begin{table}[htbp]
  \centering
  \caption{Estimated results with sample size  $5000000$}
    \begin{tabular}{ccccc}
    \toprule
          & \multicolumn{1}{l}{Algorithm 1} & \multicolumn{1}{l}{Algorithm 2} & \multicolumn{1}{l}{Algorithm 3} & \multicolumn{1}{l}{True value} \\
     \midrule
    \textbf{$\mu$} & 1.9538 & 1.9417 & 1.8405 & 2 \\
    \textbf{$\beta_1$} & 3.0152 & 3.0327 & 3.1727 & 3 \\
    \textbf{$\beta_2$} & 0.0073 & 0.0029 & -0.2881 & 0 \\
    \textbf{$\beta_3$} & -3.9460 & -3.9205 & -3.9918 & -4 \\
    \textbf{time(s)} & 80.2314    & 20.0123    & 0.5587 &  \\
     \bottomrule
    \end{tabular}%
  \label{SimulatedResults2}%
\end{table}%


%\section{Concluding remarks}


\newpage
\begin{thebibliography}{99}

\bibitem{DL2016} M. Abadi, A. Chu, I. Goodfellow, H. B. McMahan, I. Mironov, K. Talwar and L. Zhang. 2016.
 Deep learning with differential privacy.
 In Proceedings of the 2016 ACM SIGSAC Conference on Computer and Communications Security. 308-318.


\bibitem{HE1973} I. Barrodale and F. D. K. Roberts. 1973. An improved algorithm for discrete $l_1$ linear approximation.
Society for Industrial and Applied Mathematics Journal on Numerical Analysis. 839-848.

\bibitem{Tony2019} T. Cai, Y. Wang and L. Zhang. 2019. The cost of privacy: optimal rates of
convergence for parameter estimation with differential privacy. Under review: arXiv:1902.04495v3.

\bibitem{LR2009} K. Chaudhuri and C. Monteleoni. 2009. Privacy-preserving logistic regression. Proceedings of the 21st International Conference on Neural Information Processing Systems. 289-296.

\bibitem{Dwork2006}C. Dwork. 2006. Differential privacy. International colloquium on automata languages and programming. 1-12.
\bibitem{Dwork2014} C. Dwork and A. Roth. 2014. The algorithmic foundations of differential privacy.
Foundations and Trends in Theoretical Computer Science.
Vol. 9.  211-407.
\bibitem{Dwork2006C}C. Dwork, F. McSherry, K. Nissim and A. Smith. 2006. Calibrating noise to sensitivity in private data analysis. Proceedings of Theory of Cryptography Conference. 265-284
\bibitem{Google2014}   U. Erlingsson, V. Pihur and A. Korolova. 2014. Rappor: randomized
aggregatable privacy-preserving ordinal response. In Proceedings of the 2014 ACM SIGSAC
Conference on Computer and Communications Security. 1054-1067.

\bibitem{JH2007}J. Friedman, T. Hastie, H. H$\ddot{\text{o}}$fling and R. Tibshirani. 2007. Pathwise coordinate optimization. The Annals of Applied Statistics. 1(2). 302-332.

\bibitem{DR2000} D. R. Hunter and K. Lange. 2000. Quantile regression via an MM algorithm. Journal of Computational and Graphical Statistics. 9. 60-77.

\bibitem{DR2004} D. R. Hunter and K. Lange. 2004. A tutorial on MM Algorithms. The American Statistician. 58:1, 30-37
\bibitem{QR1978} R. Koenker and G. Bassett. 1978. Regression quantiles. Econometrica. 46. 33-50.
\bibitem{KD2011} D. Kifer and A. Machanavajjhala. 2011. No free lunch in data privacy. International conference on management of data. 193-204.
\bibitem{RML2001} R. Koenker and O. Geling. 2001.  Reappraising medfly longevity: a quantile regression survival analysis.
Journal of the American Statistical Association. 96. 458-468.

\bibitem{QR2001} R. Koenker and K. F. Hallock. 2001. Quantile regression. Journal of Economic Perspectives. 15. 143-156.


\bibitem{FM1993} K. Madsen and H. B. Nielsen. 1993. A finite smoothing algorithm for linear $l_1$ estimation.
SIAM Journal on Optimization. 3(2). 223-235.

\bibitem{S1973} E. J. Schlossmacher. 1973.
An iterative technique for absolute deviations curve fitting.
Journal of the American Statistical Association 68 (344). 857-859.
\bibitem{W2003}D. R. Wilson and T. Martinez. 2003. The general inefficiency of batch training for gradient descent learning. Neural Networks. 16(10). 1429-1451.
\bibitem{Wu2008} T. T. Wu and K. Lange. 2008. Coordinate descent algorithms for lasso
penalized regression.  The Annals of Applied Statistics. Vol. 2. 224-244.


\bibitem{Yang1999} S. Yang. 1999. Censored median regression using weighted empirical survival and hazard functions.
Journal of the American Statistical Association. 94. 137-145.



\end{thebibliography}
\end{document}